\documentclass{article}
\newcounter{todocounter}
\newcommand{\todonum}{\stepcounter{todocounter}{(\thetodocounter)}}

\usepackage{fullpage,color}
\usepackage{amsmath,amsthm}
\usepackage{latexsym}

\def\shownotes{1}   
\ifnum\shownotes=1
\newcommand{\authnote}[2]{{ $<<$\textsf{\footnotesize \todonum\  #1 notes:  #2}$>>$}}
\else
\newcommand{\authnote}[2]{}
\fi

\newcommand{\IA}[1]{} 
\newcommand{\DM}[1]{} 

\newcommand{\Thelma}{Thelma}
\newcommand{\Velma}{Velma}
\newcommand{\Zelma}{Zelma}

\newtheorem{claim}{Claim}

\mathchardef\mhyphen="2D 

\title{Revisiting Fast Practical Byzantine Fault Tolerance: \\
\Thelma, \Velma, and \Zelma}

\author{Ittai Abraham, Guy Gueta, Dahlia Malkhi\\
VMware Research\\
\mbox{}\\
Jean-Philippe Martin\\
Verily
}

\begin{document}

\maketitle

\newcommand{\ignore}[1]{}

\begin{abstract}
In a previous note~\cite{revisit-BFT}, we observed a safety violation in
Zyzzyva~\cite{K07,K08,KAD09} and a liveness violation in FaB~\cite{MA05,MA06}. 
In this manuscript, 
we sketch fixes to both. The same view-change core is applied in the two 
schemes, and additionally, applied to combine them and create a single, enhanced scheme
that has the benefits of both approaches. 
\end{abstract}

\section{Introduction} \label{sec:intro}
The crux of a view-change protocol is a mechanism that guarantees that 
a decision in a new view does not conflict 
with a decision that can ever be committed in any lower view. 
In~\cite{revisit-BFT}, we exposed safety issues with the view-change mechanism
of Zyzzyva~\cite{K07,K08,KAD09}, and liveness issues
with that of FaB~\cite{MA05,MA06}. 

In this manuscript, we sketch fixes for both. 
The principles we provide concentrate around a core view-change scheme that is
applied in the two schemes.

The difficulty in protocols like FaB and Zyzzyva is that they combine a
fast-track decision with a recovery-track. Therefore, a possible decision is
tranfered across views in two ways, corresponding to the two tracks, and
combining them requires care.
Simply put, in our approach a replica accepts a leader proposal in a new view as
\emph{safe} only if it is compatible
with a potential decision from the highest view-number of any lower view. 

We first sketch in \S\ref{sec:fab} a solution modeled after FaB, that we name
\Thelma{}, for a single-shot consensus.
Borrowing from FaB, \Thelma{}
provides an optimistically fast BFT solution 
in a fault model parameterized by $n=3f+2t+1$ with the following
guarantees. It is fast during periods
of synchrony, and in face of up to $t$ non-leader failures. It is always safe
against $f$ Byzantine failures.

We proceed in \S\ref{sec:ZZ} with a solution modeled after Zyzzyva, that we name
\Velma{}, for
state-machine-replication. Borrowing from Zyzzyva, \Velma{} provides an
optimistically fast execution track in a fault model of $n=3f+1$ 
with the following guarantees. It reaches a commit decision on client
requests in three-hops during periods of
synchrony and of no failures. It is always safe against up to $f$ Byzantine
failures.

Both \Thelma{} and \Velma{} require replicas to maintain
information from a constant-bounded number of previous views, and to send
constant-bounded information in new-view messages. This improved on two
previous solution frameworks that have optimistically tracks:
Refined-Quorum-Systems~\cite{RQS}, a fast single-shot Byzantine
consensus, and Azyzzyva~\cite{sevenBFT}, a fast State-Machine-Replication. In
both of these previous works, replicas maintain/send
information from all past views.

Finally, we combine in \S\ref{sec:zelma} the benefits of the parameterized fault model $n=3f+2t+1$
with full state-replication in a solution, that we name \Zelma{}. In \Zelma, a
decision on a client request is committed in the fast track during periods of
synchorny, when up to $t$ non-leader replicas are faulty. \Zelma{} provides
safety at all times against up to $f$ Byzantine failures. It guarantess liveness during
periods of synchrony with up to $f$ failures.

In all three protocols, we shed light on correctness via a
proof sketch. Formal algorithm descriptions and correctness proofs are deferred
to a future manuscript.

\subsection{Preliminaries} \label{sec:prelims}
The focus of this work is providing state-machine-replication (SMR) 
for $n$ \textit{replicas}, $f$ of which can be Byzantine faulty. An unbounded
set of \textit{clients} may form \textit{requests} and submit them to replicas.
We refer to members of the system, replicas or clients, as \textit{nodes}. 
The communication among nodes is authenticated, reliable, but asynchronous; that is, we assume that a message sent
from a correct node to another correct node is signed and eventually arrives. 

At the core of SMR is a protocol for deciding on a growing log of operation requests by clients, satisfying the following properties: 

\begin{description}

\item[Agreement] If two correct replicas commit decisions at log position
 $s$, then the decisions are the same.

\item[Validity] If a correct replica commits a decision at some log position,
then it was requested (and signed) by some client.

\item[Liveness] If some correct client submits a request, and the system is
eventually partially-synchronous~\cite{DLS_jacm88}, then eventually the replicas
commit some decision.

\end{description}

In the case of \Thelma{} (as in FaB), we
concentrate only on the core consensus problem for a single decision.

\subsection*{View Change}

The solutions we discuss employ a classical framework that revolves around an
explicit ranking among proposals via \emph{view} numbers.

Replicas all start with an initial view, and progress from one view to the next.
They accept requests and respond to messages only in their current view.

In each view there is a single designated \textit{leader}.  
In a view, zero or more decisions may be reached.
This strategy separates safety from liveness: It maintains safety even if the
system exhibits arbitrary communication delays and again up to $f$ Byzantine
failures; it provides progress during periods of synchrony.

If a sufficient number of replicas suspect that the leader is faulty, then a
view change occurs and a new leader is elected.
The mechanism to trigger moving to a higher view is of no significance for
safety, but it is crucial for liveness. On the one hand, replicas must not be
stuck in a view without progress; on the other hand, they must not move to a
higher view capriciously, preventing any view from making progress. 
Hence, a replica moves to a higher view if either
a local timer expires, or if it receives new view suggestions from
$f+1$ replicas. Liveness relies on having a constant
fraction of the views with a correct leader, whose communication with correct replicas is timely, 
thus preventing $f+1$ replicas from expiring.

Dealing with leader replacement is the pinnacle of both safety and liveness. 
A core aspect in forming agreement against failures is the need for new leaders
to safely adopt previous leader values. The reason is simple, it could be that a
previous leader has committed a decision, so the only safe thing to do is adopt
his value. 

In the prevailing solutions for the benign settings (DLS~\cite{DLS_jacm88},
Paxos~\cite{paxos}, VR~\cite{oki_podc88}, Raft~\cite{ongaro_atc14}), leader
replacement is done by reading from a view-change quorum of $n-f$ replicas, and
choosing the value with the maximal view \footnote{In DLS, the term
\textit{phase} is used, and in Paxos, \textit{ballot}.} number.
Note that
$n-f$ captures a requirement that the quorum intersects every leader-quorum in previous views (not only the most recent one). It is crucial to take into consideration how leader quorums of multiple previous views interplay.
Choosing the value with the maximal view-number is crucial because there may be
multiple conflicting values and choosing an arbitrarily value is not always a
safe decision.


A similar paradigm holds in PBFT~\cite{CL99,CL02}. The new leader needs to read
from a \emph{view-change quorum} of $n-f$
replicas and choose a value with the maximal view-number.
Different from the benign case, in
the Byzantine settings, uniqueness is achieved by using enlarged, Byzantine
quorums~\cite{MR_dc98}. Byzantine quorums guarantee intersection not just in any
node but in a correct node. 

In Byzantine settings, a correct node also needs to prove a decision value to a new leader.
This is done in PBFT\footnote{We refer here to the PBFT version with signed
messages~\cite{CL99}.} by adding another phase before a decision.  
The first phase ensure uniqueness via \emph{prepare} messages from $n-f$ nodes.
In the second phase, nodes send a \textit{commit-certificate} consisting of
$n-f$ prepare messages. A decision can be reached when a \emph{commit-quorum} of $n-f$ nodes have sent a commit-certificate.  

The two-phase scheme guarantees the follows. If there is a decision, there
exists a correct
node in the intersection between a commit-quorum and a view-change quorum that passes a
commit-certificate to the next view. 

Indeed, a new leader chooses in PBFT a value whose commit-certificate, rather
than a prepare, has the maximal view-number.

\newpage
\section{\Thelma{}: Revisiting the FaB View-Change} \label{sec:fab}
\subsection{A Skeletal Overview of PFaB}

Martin and Alvisi introduce Fast Byzantine Consensus (FaB) 
in~\cite{MA05, MA06}, a family of protocols parameterized by various resilience assumptions.
The papers use the Paxos terminology to model roles: \textit{proposers}, \textit{acceptors}, and
\textit{learners}. And it employs \textit{proposal numbers} to enumerate
proposals. We will adhere to the Zyzzyva (and PBFT) terminology, and translate
those to \emph{leaders}, \emph{replicas}, and \emph{view-numbers}.

FaB has two variants.  The first FaB variant works with $n=5f+1$ replicas, 
trading fast termination by with reduced resilience. 
Here, we focus on the second variant, parameterized with
$n=3f+2t+1$, where $t \le f$. We refer to it here as PFaB.
It works in two tracks, a fast track and a recovery track.  

The fast track protocol of PFaB is an easy two-step protocol.
A leader pre-proposes a value to replicas, who each \textit{accept} one value
per view and respond with a \emph{prepare} message. A decision is reached in
PFaB when a \emph{fast-quorum} of $n-t$ replicas accept the leader's proposal and send a prepare response for it. 

The fast track is guaranteed to complete in periods of synchrony with a correct
leader and up to $t$ Byzantine replicas.  However,
parameterized FaB does not necessarily guarantee fast progress
even in periods of synchrony, if the parameter $t$ threshold of failures is exceeded. 
That is, although PFaB is always safe despite up to $f$ Byzantine failures,
it is not always fast.  

If progress is stalled, PFaB allows progress via a recovery protocol, which is
essentially PBFT (adapted to $n=3f+2t+1$). 
The recovery track is guaranteed to complete during periods of synchrony
if the number of actual Byzantine failures does not exceed $f$.

More precisely, in PFaB, the recovery track revolves around forming a
commit-certificate called a \textit{commit-proof}.  When replicas accept a
leader proposal, in addition to
sending \textit{prepare} messages ({\tt ACCEPTED}) to the leader, replicas also
send signed prepare messages to each other.  We say that a replica has a
commit-certificate for a value $v$ if it receives in a view prepare messages for $v$ 
from a \emph{recovery-quorum} of $(n-f-t)$ replicas.
Upon obtaining a commit-certificate, a replica
sends it in a \textit{commit} message ({\tt COMMITPROOF}) to other replicas.  

A decision is reached if either a fast-quorum of $n-t$ replicas send prepare messages (for the same
value), or a recovery-quorum of $(n-f-t)$ replicas send commit messages (for the same value).

The core mechanism in PFaB for transferring safe values across views is a
\emph{progress certificate} containing new-view messages ({\tt REP}) from a
\emph{progress-quorum} of $n-f$ replicas. 
A new-view message from a replica contains the new view's
number,
the last value it sent in a prepare message, and the last
commit-certificate it sent in a commit message. 

In PFaB, a progress-certificate for a specific new view is said to \textit{vouch for} a value $v$ if there does not exist
a set of $f+t+1$ new-view messages with an identical prepare value $v'$ such that $v' \neq v$; and
there does not exist any commit-certificate with value $v'$ such that $v' \neq v$.

\subsection{\Thelma}

We now outline a new view-change scheme within the above PFaB protocol
framework. We will refer to the fixed protocol as \Thelma. 

In order to fix PFaB, each replica needs to maintain 
with the last prepare and commit messages it sent
their original view numbers.  When a replica 
copies its last prepare and commit messages into
a new-view message, it needs to attach 
the original view numbers to them.

A decision is transferred across views via a progress-certificate as follows. 

\begin{itemize}

\item 
A possible fast-track decision is transferred across views via
a set of prepares intersecting a progress-quorum. Complicating
matters,  each prepare may be repeated in
higher views, hence different prepares in the intersection may carry different view-numbers. 

\item
A possible recovery track decision
is transferred across views via a commit-certificate.

\end{itemize}

To combine possible decision values from both tracks, replicas
need to choose the highest previous view in which a
decision is possible. If both tracks appear possible for the same view-number,
then a commit-certificate provides evidence against a
potential fast-track decision in the same view.

More specifically,
let $P$ be a progress-certificate consisting
of view-change messages from a progress-quorum of $n-f$ replicas. 

In order to simplify processing $P$, we introduce several key notions.

\begin{description}

\item[fast-certificate($d$):] The highest view-number $v$ such that $f+t+1$ prepare messages in
$P$ contain the value $d$ and a view-number at least $v$. If no
such $v$ exists, we set fast-certificate($d$) to $-1$.

\item[PFAST($d$):]
If fast-certificate($d$) has the highest view-number among fast-certificates in
$P$, then PFAST($d$) is true; otherwise, it is false.

\item[PSLOW($d$):]
If a commit-certificate for $d$ exists in $P$ and has the highest view-number
among commit-certificates, then PSLOW($d$) is true; otherwise, it is false.

\end{description}

We are now ready to determine when a value $d$ is safe for a progress-certificate $P$:

\begin{enumerate}

    \item PSLOW($d$) holds, and for all $d'$, fast-certificate($d'$) has
view-number no higher than the commit-certificate for $d$, or

    \item PFAST($d$) holds, and for all $d'$, a commit-certificate for $d'$ if
exists has view-number lower than fast-certificate($d$), or

    \item for no value $d'$ does PSLOW($d'$) or PFAST($d'$) hold; hence, all
values are safe.

\end{enumerate}

\subsection{Examples}

To demonstrate \Thelma{}'s view-change, we revisit the ``stuck'' scenario
in~\cite{revisit-BFT}, as well as another scenario.

For these scenarios, we set $f=1$, $t=0$, $n=3f+2t+1=4$.
Denote the replicas by $i_1$, $i_2$, $i_3$, $i_4$, one of whom, say $i_1$, is Byzantine.

\medskip
\noindent
The first scenario goes through one view change.

\paragraph{View 1:}

  \begin{enumerate}
  \item Leader $i_1$ (Byzantine) pre-proposes value $d$ to $i_2$, $i_3$.

  \item $i_1$, $i_2$, $i_3$ send prepare messages for $d$.

  \item $i_2$ collects a view-1 commit-certificate for $d$ and sends a commit message
for $d$. 

  \item Meanwhile, the leader $i_1$ equivocates and pre-proposes $d'$ to $i_4$. 
  
  \end{enumerate}

\paragraph{View 2:}

  \begin{enumerate}

  \item The new leader $i_2$ collects a progress-certificate consisting of new-view messages from a quorum of $3$ replicas (including itself):

	\begin{itemize}
                \item from $i_1$, the new-view message contains a prepare for
$d'$ from view 1, and no commit-certificate.
                \item from $i_2$, the new-view message contains a prepare for
$d$ from view 1, and a view-1 commit-certificate for it.
                \item from $i_4$, the new-view message contains a prepare for
$d'$, and no commit-certificate.
     \end{itemize}
	
  \end{enumerate}

In this progress-certificate, PSLOW($d$) holds, and no fast-certificate has a
view-number higher than $d$'s commit-certificate view-number ($1$).
Therefore, it determines $d$ as the only safe value to propose. 

Indeed, notice that a decision on $d$ is still possible in view 1: $i_3$ may send a
commit message, and $i_1$ (Byzantine) may send a commit message even though it
already moved to view 2.

\hrule
\hrule
\medskip
\noindent
The second scenario goes through two view changes.

\paragraph{View 1:}

  \begin{enumerate}
  \item Leader $i_1$ (Byzantine) pre-proposes value $d$ to $i_2$, $i_3$.

  \item $i_1$, $i_2$, $i_3$ send prepare messages for $d$.

  \item $i_1$ collects a commit-certificate for $d$ (and stalls). 

  \item Meanwhile, the leader $i_1$ equivocates and pre-proposes $d'$ to $i_4$. 
  
  \end{enumerate}

\paragraph{View 2:}

  \begin{enumerate}

  \item The new leader $i_2$ collects a progress-certificate consisting of new-view messages from a quorum of $3$ replicas (including itself):

	\begin{itemize}
                \item from $i_1$, a new-view message contains the prepare for
$d'$ from view 1, and no commit-certificate.
                \item from $i_2$, a new-view message contains the prepare for
$d$ from view 1, and no commit-certificate.
                \item from $i_4$, a new-view message contains the prepare for $d'$, and no commit-certificate.
     \end{itemize}

  \item $d'$ is a safe value for the progress-certificate since PFAST($d'$) it
true, and there are no commit-certificates.  $i_2$ uses the progress-certificate to pre-propose $d'$ 
to replicas as a safe value.

  \item everyone sends prepare messages for $d'$, and a client learns that $d'$ is
committed.
	
  \end{enumerate}

\paragraph{View 3:}

  \begin{enumerate}

   \item The new leader $i_3$ collects a progress-certificate consisting of
new-view messages from a quorum of $3$ replicas: 

	\begin{itemize}

                \item from $i_1$, the new-view message hides the fact that it
prepared a value in view $2$, and contains a prepare for $d$ from view 1, and
a view-1 commit-certificate for it. 

                \item from $i_3$ and $i_4$, the new-view message contains a
prepare for $d'$ from view $2$ 
	\end{itemize}

\end{enumerate}

In this progress-certificate, PFAST($d'$) holds, and the highest
commit-certificate has view-number $1$ (for $d$).
Therefore, it determines $B$ as the only safe value to propose. 

Note that this scenario demonstrates that a commit-certificate may not
necessarily override a set of $f+1$ prepares, unless its view-number is at
least that of the highest fast-certificate.

\subsection{Correctness}

\begin{claim}
Let a value $d$ (ever) become a committed decision in view $v$ in the fast track.
Then the progress-certificate for every higher view $v' > v$ determines $d$ as
the only safe value.
\end{claim}

\begin{proof}[Proof Sketch]
Since $d$ becomes a committed decision in view $v$ in the fast track, there is a
fast-quorum $Q$ of $n-t$ replicas that send prepare messages for $d$ in view $v$, 
before moving to any higher view. 

By way of contradiction, let $P'$ be the progress-certificate whose view-number
$v' > v$ is the lowest, such that $d$ is not the only safe value for $P'$.
We are going to draw certain conclusions about PFAST and PSLOW for $P'$ in order to
arrive at a contradiction.

\paragraph{PFAST for $P'$.}
First, let us compute PFAST for $P'$.

Denote by $Q'$ the progress-quorum of $n-f$ replicas whose new-view messages are included in
$P'$. 
$Q$ and $Q'$ intersect in a set of at least $f+t+1$ correct replicas. These
replicas report the prepare messages they sent in a view $v$ or higher (up to $v'-1$). 
By assumption, these prepares all contain the value $d$. Hence, fast-certificate($d$)
is at least $v$ in $P'$. 

By assumption, no value $d' \neq d$ is safe to propose in a view higher than $v$
and less than $v'$.
Hence, the number of replicas with prepare messages in $P'$ for any value $d'
\neq d$ with view $v$ or higher is at most $f+t$. Specifically, in view $v$, 
there may be at most $t$ correct replicas outside $Q$ that have prepares for $d'$.
Additionally, there may be $f$ Byzantine replicas with prepares for $d'$ with
arbitrary view numbers. 
In total, there are not enough prepares for fast-certificate($d'$) to be
$v$ or higher. 

We conclude that PFAST($d$) is true, and for every other $d'$, PFAST($d'$) is
false.

\paragraph{PSLOW for $P'$.}
We now compute PSLOW for $P'$. Once again,
we already showed that the number of replicas with prepare messages 
for any value $d' \neq d$ 
whose view is $v$ or higher is at most $f+t$. 
In total, there are not enough prepares for 
a commit-certificate on $d'$ to have view $v$ or higher.

Therefore, either PSLOW($d'$) is false, or its commit-certificate has
view-number lower than $v$.

\medskip

\noindent
Putting the constraints on PFAST and PSLOW for $P'$ together, we conclude that
$d$ is the only safe value for $P'$, and we arrive at a contradiction.
\end{proof}

\begin{claim}
Let a value $d$ (ever) become a committed decision in view $v$ in the recovery track.
Then the progress-certificate for every higher view $v' > v$ determines $d$ as the
safe value.
\end{claim}

\begin{proof}[Proof Sketch]
Since $d$ becomes a committed decision in view $v$ in the slow track, there is a
recovery-quorum
$Q$ of $n-f-t$ replicas that send commit messages for $d$ in view $v$, 
before moving to any higher view. 

By way of contradiction, let $P'$ be the progress-certificate whose view-number
$v' > v$ is the lowest, such that $d$ is not the (only) safe value for $P'$.
We are going to draw certain conclusions about PFAST and PSLOW for $P'$ in order to
arrive at a contradiction.

\paragraph{PFAST for $P'$.}
First, let us compute PFAST for $P'$.

Denote by $Q'$ the progress-quorum of $n-f$ replicas whose new-view messages are included in
$P'$. 

By assumption, no value $d' \neq d$ is safe to propose in a view higher than $v$
and less than $v'$.
However, there may be $f$ Byzantine replicas in $Q'$ with prepares for $d'$ with
arbitrary view numbers. 
There may be additionally up to $f+t$ correct replicas outside $Q$ that have prepares for $d'$ in view $v$.
Therefore, fast-certificate($d'$), if non-negative, can be at most $v$.
We conclude that either PFAST($d'$) is false,  or fast-certificate($d'$) has view number at most $v$,
or both.

\paragraph{PSLOW for $P'$.}
We now compute PSLOW for $P'$. 
We already showed that the number of replicas with prepare messages in view $v$
or higher for any value $d' \neq d$ is at most $2f+t$. 
In total, there are not enough prepares for 
a commit-certificate on $d'$ to have a view number $v$ or higher.

On the other hand, $Q$ and $Q'$ intersect in a set of at least $t+1$ correct replicas. These
replicas report the commit message they sent in view $v$ or higher (up to $v'-1$). 
By assumption, these commits all contain the value $d$. Hence,
PSLOW($d$) holds, and for no other $d'$ is PSLOW($d'$) true.

\medskip

\noindent
Putting the constraints on PFAST and PSLOW for $P'$ together, we conclude that
$d$ is the only safe value for $P'$, and we again arrive at a contradiction.
\end{proof}

\begin{claim}
Every progress certificate determine some safe value.
\end{claim}

\begin{proof}[Proof Sketch]
Since the rules for determining the safe value for a certificate are all
positive, i.e., no values are explicitly ruled out by a certificate, there is
always a possible safe value for every progress certificate. 
\end{proof}

\newpage
\section{\Velma{}: Revisiting the Zyzzyva View-Change} \label{sec:ZZ}
\subsection{A Skeletal Overview of Zyzzyva}

Zyzzyva~\cite{K07,K08,KAD09} is a full State-Machine-Replication (SMR) protocol that has two commit paths. 
A two-phase path that resembles PBFT and a fast path. 

The fast path does not have commit messages, 
and replicas speculatively execute requests and optimistically return
\emph{prepare} results directly to clients.  A client learns a commit decision in the fast path by
seeing $3f+1$ prepare messages.
The optimistic mode is coupled with a recovery mode that guarantees progress in
face of failures. The recovery mode intertwines a two-phase ($2f+1$)-quorum exchange into
the protocol.
In the two-phase recovery mode, replicas proceed to speculatively execute commands as well.
In both modes, replicas may need to roll back speculative executions if in the end
they conflict with committed decisions.

In Zyzzyva, a possible decision value is transferred across views in two
possible ways, corresponding to the two decision tracks of the protocol (fast
and two-phase):
In the fast track, a possible decision value manifests itself as $f+1$ prepare
messages. 
In the two-phase track, it manifests itself as a commit-certificate. 
Combining the two, Zyzzyva prefers a commit-certificate over $f+1$
prepares; and among two commit-certificates, it prefers the one with the longer
request-log.  

We proceed with a skeletal description of the Zyzzyva
sub-protocols, a fast-track sub-protocol, a two-phase sub-protocol, and a
view-change sub-protocol. 
Our description omits details regarding checkpoint management, and many other optimizations, which are not crucial for correctness considerations, 
and are described in the original papers.

\paragraph{Messages.}

All messages in the protocol are signed and may be forwarded carrying the original sender's signature. 
The protocol makes use of the following interactions.

\begin{description}

\item{Client-request:}
A \textit{client-request} (REQUEST) from a client to the leader contains some
operation $o$, whose semantics are completely opaque for the purpose of this
discussion. 

\item{Ordering-request:}
A leader's \emph{pre-prepare} message is called an \textit{ordering-request}
(ORDER-REQ), and contains a leader's
log of client requests $OR_n = (o_1, ..., o_n)$. 
(In practice, the leader sends only the last request and a hash of the history of prior 
operations; a node can request the leader to re-send any missing 
operations.)

\item{Ordering-response:}
When a replica \textit{accepts} a valid pre-prepare request, it speculatively
executes it and sends the result in 
a \emph{prepare} message called an \emph{ordering-response} (SPEC-RESPONSE). 

\item{Commit-request:}
A \textit{commit-request} (COMMIT) from the client to the replicas includes a \textit{commit-certificate} $CC$, 
a set of $2f+1$ signed replica responses (SPEC-RESPONSE) to an
(identical) ordering-request $OR_n$.

\item{Commit-response:}
When a replica obtains a valid commit-certificate $CC$ for $OR_n$, it responds
to client requests in $OR_n$
with a \textit{commit} message called a \emph{commit-response} (LOCAL-COMMIT).

\item{View-change:}
A \textit{view-change} (VIEW-CHANGE) message from a replica to the leader of a
new view captures the replica's \emph{local state}. 

\item{New-view:}
A \textit{new-view} (NEW-VIEW) message from the leader of a new view contains a
set $P$ of view-change messages the leader collected, which serves as a
\emph{leader-proof}. It includes an ordering-request for a \emph{leader-log} $G_n = (o_1, ..., o_n)$. 

\end{description}

\paragraph{The fast-track sub-protocol.}

Zyzzyva contains a fast-track protocol in which a client learns the result of a
request in
only three message latencies, and only a linear number of crypto operations. 
It works as follows.

A client sends a request $o$ to the current leader. The
current leader extends its local log with the request $o$ to $OR_n$, and sends a
pre-prepare (ordering-request) carrying $OR_n$. We did not say how a leader's local
log is initialized. Below we discuss the protocol for a leader to pick an
initial log when starting a new view. 

A replica \textit{accepts} a pre-prepare from the leader of the current
view if it has valid format, and it extends
any previous pre-prepare from this leader. 
Upon accepting a pre-prepare, a replica extends its local log to $OR_n$ 
It speculatively executes it, 
and sends the result directly to the client in a \emph{prepare} message.

A decision is reached on $OR_n$ in view $v$ in the fast track 
when $3f+1$ distinct replicas have sent a prepare message for it. 

\paragraph{The two-phase sub-protocol.}

If progress is stalled, then a client waits to collect a
\emph{commit-certificate}, a set of $2f+1$ prepare responses for 
$OR_n$. 
Then the client sends a commit-request carrying the commit-certificate to the
replicas. A replica responds to a valid commit-request with a
\emph{commit} message.

A decision is reached on $OR_n$ in view $v$ in the two-phase
track when $2f+1$ distinct replica have sent a commit message for it.

\paragraph{The view-change protocol.}

The core mechanism in Zyzzyva for transferring safe values across views is
for a new Zyzzyva leader to collect a set $P$ of
view-change messages from a quorum of $2f+1$ replicas. Each replica sends a
view-change message containing the replica's \textit{local state}:
Its local request-log,
and the commit-certificate with the highest view number it responded to with a
commit message, if any.

The leader processes the set $P$ as follows.

\begin{enumerate}
 \item
  Initially, it sets a \emph{leader-log} $G$ to an empty log.
  
  \item
  If any view-change message contains a valid commit-certificate, then it
selects the one with the longest request-log $OR_n$ and copies $OR_n$ to $G$.  
 
\item
If $f+1$ view-change messages contain the same request-log $OR'_m$,
   then it extends the tail of $G$ with requests from $OR'_m$. (If there are two
$OR'_m$ logs satisfying this, one is selected arbitrarily.)
   
\item
Finally, 
it pads $G$ with null request entries up to the length of the longest log of any
valid prepare.

\end{enumerate}

The leader sends a new-view message to all the replica. The message
includes the new view number $v+1$, the set $P$ of view-change messages the
leader collected as a \emph{leader-proof} for view $(v+1)$, and the \emph{leader-log}
$G$. A replica accepts a new-view
message if it is valid, and \emph{adopts} the leader log. It may need to roll
back speculatively executed requests, and process new ones.


\subsection{\Velma}

We now outline a new view-change scheme within the above Zyzzyva protocol
framework. We will
refer to the fixed protocol as \Velma.

In order to fix Zyzzyva, we change the method for the leader and for replicas to
select safe leader-logs during a view-change. Let
$P$ be a set of view-change messages from a view-change quorum of $2f+1$ replicas.

In order to simplify processing $P$, we introduce several key notions.
These are similar, but not identical to those introduced in \Thelma, because in \Velma{} we need to determine safety of request-logs, rather than of a single value.

\begin{description}

\item[extends:] 
An \emph{extends} relation between two request-logs $O_1$, $O_2$,
denoted $O_1 \sqsubseteq O_2$, indicates that $O_1$ is a prefix (not
necessarily strict) of $O_2$. If $O_1 \not\sqsubseteq O_2$ and $O_2 \not\sqsubseteq O_1$
then they are \emph{conflicting}. 

\item[fast-certificate($O$):]
\DM{Is this clearer?}
The highest view-number $v$ such that $f+1$ prepare messages in $P$ contain a
log that extends $O$ and a view-number at least $v$.

To explain this notion, recall that a
fast-track decision on a log $O$ in some view $v$ intersects a view-change quorum in $f+1$ correct replicas. 
However, since a committed decision may be repeatedly proposed, and possibly
extended, in higher views, each of these $f+1$ replicas sends a new-view message
with a prepare containing a log that may extend $O$, and may have a view-number $v$ or
higher. 

If no such view-number exists, we set fast-certificate($O$) to $-1$.

\item[slow-certificate($O$):]
The highest view-number $v$ for which a commit-certificate exists for $O$.

If no such view-number exists, we set slow-certificate($O$) to $-1$.

\item[PFAST($O$):]
\DM{Is this clearer?}
If fast-certificate($O$) has the highest view-number among fast-certificates in
$P$, and there is no $O'$ extending $O$ (i.e. $O \sqsubset O'$) with the same
fast-certificate, then PFAST($O$) is true; otherwise, it is false.

\item[PSLOW($O$):]
If slow-certificate($O$) has the highest view-number
among slow-certificates, and there is no $O'$ extending $O$ (i.e. $O \sqsubset
O'$) with the same slow-certificate, then PSLOW($O$) is true; otherwise it is false.

\end{description}

We are now ready to determine when a log $O$ is safe for a progress-certificate $P$:

\begin{enumerate}

    \item PSLOW($O$) holds, and for all $O'$, fast-certificate($O'$) is
lower than slow-certificate($O$), or

    \item PFAST($O$) holds, and for all $O'$, slow-certificate($O'$)
is lower than fast-certificate($O$), or

    \item PSLOW($O$) holds, and for every $O'$ whose 
fast-certificate($O'$) has the same view-number as slow-certificate($O$), we have $O' \sqsubseteq O$, or

    \item PFAST($O$) holds, and for every $O'$ whose 
fast-certificate($O'$) has the same view-number as slow-certificate($O$), we have $O' \sqsubseteq O$, or

    \item for no value $d'$ does PSLOW($d$) or PFAST($d$) hold; hence, all
values are safe.

\end{enumerate}

\subsection{Examples}

To demonstrate \Velma{}'s view-change, we revisit the two safety-violation scenarios
in~\cite{revisit-BFT}.

Our first scenario requires four replicas $i_1$, $i_2$, $i_3$, $i_4$, of which one, $i_1$, is Byzantine. 
It proceeds in $3$ views, and arrives at a conflicting decision on the first log position.

\paragraph{View 1: Creating a commit-certificate for $(a)$.}

\begin{enumerate}
\item Leader $i_1$ sends pre-prepare with log $(a)$ to replicas $i_2$ and
$i_3$.

\item Leader $i_1$ (Byzantine) equivocates and sends pre-prepare with log
$(b)$ to replica $i_4$.

\item
Replicas $i_2$, $i_3$ speculatively execute $a$, obtain a speculative result and send it in a
prepare message to a client.

\item The client collects a commit-certificate $cert$ of view-1 prepares from $i_1$, $i_2$, $i_3$
for the log $(a)$ and sends it to $i_1$.

\end{enumerate}

\paragraph{View 2: Deciding $(b)$.}

\begin{enumerate}

\item The new leader $i_2$ collects view-change messages from a quorum of $3$
(including itself) as follow: 
  \begin{itemize}
  \item Replica $i_2$ sends its view-1 prepare for log $(a)$.
 
  \item Replica $i_4$ sends its view-1 prepare for log $(b)$:

  \item Replica $i_1$ (which is Byzantine) joins $i_4$ and sends a view-1
prepare for log $(b)$. 

  \end{itemize}  
  
Based on these view-change messages, PFAST is $(1, (b))$ and PSLOW is $(-1,
\bot)$. Hence, $(b)$ is the only safe choice.

$i_2$ sends a new-view message consisting of the log $(b)$, using
the set of view-change messages as proof that this is a safe value.

\item
Every replica zeros its log (undoing $a$, if needed), speculatively execute
$b$, and sends a view-2 prepare for log $(b)$.

\item 
A client collects speculative-responses from all replicas, and 
$b$ becomes successfully committed at log position $1$.

\end{enumerate}

\paragraph{View 3: Choosing the right commit-certificate.}

\begin{enumerate}

\item The new leader $i_3$ collects view-change messages from a quorum of $s$ as
follow:

  \begin{itemize}
  \item Replica $i_1$, which is Byzantine, hides the value it prepared in view
$2$, and sends a view-1 commit-certificate $cert$ (see above) for $(a)$.

  \item Replicas $i_3$ and $i_4$ send their view-2 prepares for log $(b)$.

  \end{itemize}
    
Based on these view-change messages, PFAST is $(2, (b))$, PSLOW is $(1, (a))$,
and $(b)$ is the only safe choice.

\end{enumerate}


\medskip
\noindent
The second scenario is also rather short, uses four replicas, and two view
changes.

\paragraph{View 1: Creating a commit-certificate for $(a_1, a_2)$.}

\begin{enumerate}
\item Leader $i_1$ sends pre-prepare with log $(a_1,a_2)$ to replicas $i_2$ and
$i_3$.

\item Leader $i_1$ (Byzantine) equivocates and sends pre-prepare with log
$(b_1,b_2)$ to replica $i_4$.

\item
Replicas $i_2$, $i_3$ speculatively execute $a_1$ followed by $a_2$, obtain a speculative result and send it in a
prepare message to a client.

\item The client collects a commit-certificate $cert_1$ of view-1 prepares from $i_1$, $i_2$, $i_3$
for the log $(a_1,a_2)$ and sends it to $i_3$.

\end{enumerate}

\paragraph{View 2: Deciding $(b_1)$.}

\begin{enumerate}

\item The new leader $i_2$ collects view-change messages from a quorum of $3$
(including itself) as follow: 
  \begin{itemize}
  \item Replica $i_2$ sends its view-1 prepare for log $(a_1,a_2)$.
 
  \item Replica $i_4$ sends its view-1 prepare for log $(b_1, b_2)$.

  \item Replica $i_1$ (which is Byzantine) joins $i_4$ and sends a view-1
prepapre for log $(b_1,b_2)$. 
      
  \end{itemize}  
  
Based on these view-change messages, PFAST is $(1, (b_1,b_2))$i, and PSLOW is $(-1,
\bot)$. Hence, $(b_1,b_2)$ is the only safe choice.

$i_2$ sends a new-view message consisting of the log $(b_1,b_2)$, using
the set of view-change messages as proof that this is a safe value.

\item
Every replica zeros its log (undoing $a_1$,$a_2$, if needed). It first proceeds
to speculatively execute $b_1$, and sends a view-2 prepare for log $(b_1)$.

\item 
A client collects a commit-certificate $cert_2$ of view-1 prepares from $i_1$,
$i_2$,$i_4$ and send it to replicas  
$i_1$, $i_2$, and $i_4$. They respond with a commit message for log $(b_1)$. 

\item 
The client collects commit messages and $b_1$ becomes successfully
committed at log position $\mathbf{1}$.

\end{enumerate}

\paragraph{View 3: Choosing the right commit-certificate.}

\begin{enumerate}

\item The new leader $i_3$ collects view-change messages from a quorum of $s$ as
follow:

  \begin{itemize}
  \item Replica $i_3$ sends commit-certificate $cert_1$ (see above) for $(a_1, a_2)$.

  \item Replica $i_4$ sends commit-certificate $cert_2$ (see above) for $(b_1)$,
and its local log $(b_1, b_2)$.

  \item Replica $i_1$ (Byzantine) can join either one, or even send an
view-change message with an empty log. 
  \end{itemize}
    
Based on these view-change messages, PFAST is $(2, (b_1))$, and PSLOW has
view-number at most $2$. Therefore, $(b_1)$ is the only safe choice for log
position $1$.

\end{enumerate}

\subsection{Correctness}

The correctness argument for \Velma{} are similar in essence to \Thelma.
However, care must be taken to preserve consistency across a sequence of consensus
decisions, rather than one. And each decision must wait for execution results of a log of requests,
rather than commit to an individual proposal value. 

\begin{claim}
Let a request-log $OR_n$ (ever) become a committed decision in view $v$ in the fast track.
If a leader-proof of a higher view $v' > v$ determines $OR$ as a safe
leader-log, then $OR_n \sqsubseteq OR$.
\end{claim}

\begin{proof}[Proof Sketch]
Since $OR_n$ becomes a committed decision in view $v$ in the fast track, there is a
fast-quorum $Q$ of $n$ replicas that send prepare messages for $OR_n$ in view $v$, 
before moving to any higher view. 

By way of contradiction, let $P'$ be the leader-proof whose view-number
$v' > v$ is the lowest, 
and $OR'$ a safe leader-log of $P'$ conflicting with $OR_n$.
We are going to draw certain conclusions about PFAST and PSLOW for $P'$ in order to
arrive at a contradiction.

\paragraph{PFAST for $P'$.}
First, let us compute PFAST for $P'$.

Denote by $Q'$ the view-change quorum of $2f+1$ replicas whose new-view messages are included in
$P'$. 
$Q$ and $Q'$ intersect in a set of at least $f+1$ correct replicas. These
replicas report the prepare messages they sent in a view $v$ or higher (up to $v'-1$). 
By assumption, these prepares do not conflict with $OR_n$.
Hence, fast-certificate($OR_n$) is at least $v$ in $P'$. 

By assumption, 
$OR'$ is not safe to propose in a view higher than $v$ and less than $v'$. Hence, the number of replicas with prepare messages in $P'$ for $OR'$ is at most $f$. 

In total, there are not enough prepares for fast-certificate($OR'$) to be $v$ or
higher. 
We conclude that PFAST($OR'$) is false for every $OR'$ conflicting with $OR_n$, and PFAST($OR'_n$) is true for some some $OR'_n$ extending $OR_n$.

\paragraph{PSLOW for $P'$.}
We now compute PSLOW for $P'$. 
We already showed that the number of prepares for $OR'$
with view-number $v$ or higher is at most $f$. 
Therefore, there are not enough prepares for 
a commit-certificate on $OR'$ to have view $v$ or higher.

Therefore, either PSLOW($OR'$) is false, or its commit-certificate has view-number lower than $v$.

\medskip

\noindent
Putting the constraints on PFAST and PSLOW for $P'$ together, we conclude that
every safe leader-log for $P'$ extends $OR_n$, and we arrive at a contradiction.
\end{proof}

\begin{claim}
Let a request-log $OR_n$ (ever) become a committed decision in view $v$ in the
two-phase track.
If a leader-proof of a higher view $v' > v$ determines $OR$ as a safe
leader-log, then $OR_n \sqsubseteq OR$.
\end{claim}

\begin{proof}[Proof Sketch]
Since $OR_n$ becomes a committed decision in view $v$ in the slow track, there is a
two-phase quorum
$Q$ of $2f+1$ replicas that send commit messages for $OR_n$ in view $v$, 
before moving to any higher view. 

By way of contradiction, let $P'$ be the leader-proof whose view-number
$v' > v$ is the lowest, 
and $OR'$ a safe leader-log of $P'$ conflicting with $OR_n$.
We are going to draw certain conclusions about PFAST and PSLOW for $P'$ in order to
arrive at a contradiction.

\paragraph{PFAST for $P'$.}
First, let us compute PFAST for $P'$.

Denote by $Q'$ the view-change quorum of $2f+1$ replicas whose new-view messages are included in
$P'$. 
By assumption, 
the only prepare messages in $P'$ whose view is higher than $v$ 
and have $OR'$ are faulty. Hence, there are at most $f$ of them. 
In view $v$, there may be additionally up to $f$ correct replicas outside $Q$
that have prepares for $OR'$.
Therefore in total, fast-certificate($OR'$) in $P'$ is at most $v$.
We conclude that either PFAST($OR'$) is false, or fast-certificate($OR'$) has view-number at most $v$,
or both.

\paragraph{PSLOW for $P'$.}
We now compute PSLOW for $P'$. 

We already showed that 
the number of prepare messages in $P'$ whose view is $v$ or higher 
and have $OR'$ is at most $2f$. 
In total, there are not enough prepares for 
a commit-certificate on $OR'$ to have a view number $v$ or higher.

On the other hand, $Q$ and $Q'$ intersect in at least one correct
replica. This replica reports the commit message it sent in view $v$ or higher (up to $v'-1$). 
By assumption, these commits all extend the value $OR_n$. 

Hence, PSLOW($OR'$) is false, and PSLOW($OR'_n$) is true for some $OR'_n$ extending $OR_n$. 

\medskip

\noindent
Putting the constraints on PFAST and PSLOW for $P'$ together, we conclude that
every safe leader-log for $P'$ extends $OR_n$, and we again arrive at a contradiction.
\end{proof}

\begin{claim}
Every leader-proof determines some safe value.
\end{claim}

\begin{proof}[Proof Sketch]
Since the rules for determining the safe value for a proof are all
positive, i.e., no values are explicitly ruled out by it, there is
always a possible safe leader-log for every leader-proof. 
\end{proof}


\newpage
\section{\Zelma{}: Putting Fab and Zyzzyva View-Change Together}
\label{sec:zelma}
Having fixed the view-change in Zyzzyva allows us to combine the mechanism for
optimistic (fast track) execution with the parameterized $n=3f+2t+1$ failure model of FaB.
We name the combined solution \Zelma.

In \Zelma, a fast-quorum consists of $n-t$ replicas. The fast track allows a
leader to extend the current log with a new client request. The client can
commit a decision by seeing $n-t$ prepare messages. The optimistic 
mode is guaranteed to complete in periods of synchrony with a correct leader and
up to $t$ Byzantine replicas. 

A two-phase quorum consists of $n-f-t$ replicas. It is essentially PBFT adapted
to the parameterized fault model. More specifically, a client collects a
commit-certificate consisting of signed prepares from a quorum of $n-f-t$
replicas. It forwards the certificate to the replicas. A decision is reached when
a quorum of $n-f-t$ replicas received the commit-certificate.

The two-phase track is guaranteed to complete is
periods of synchrony with a correct leader, up to $f$ Byzantine replicas and
additionally up to $t$ slow replicas. 

\Zelma{} maintains safety at all times against up to $f$ Byzantine failures.

If progress is stalled, \Zelma{} provides eventual progress via
a view-change protocol. A view-change quorum consists of $n-f$ replicas. 

In \Zelma, a possible decision value is transferred across views in two
possible ways, corresponding to the two decision tracks of the protocol (fast
and two-phase):

In the fast track, a possible decision value manifests itself as $f+t+1$ prepare
messages. These prepares may potentially have different view numbers and
request-logs, but they all contain the committed decision as prefix. A
fast-certificate records the $(f+t+1)$-highest view-number among all the prepares for the
same request-log prefix. 

In the two-phase track, a decision manifests itself as a commit-certificate with
$2f+t+1$ identical prepares. 

Combining the two, \Zelma{} picks the highest view 
with either a commit-certificate or a fast-certificate. This becomes the
initial leader-log for a new-view. 
If the highest commit-certificate and fast-certificate have equal
view-numbers, then the initial leader-log is a concatenation of the
commit-certificate log, with any remaining entries from the fast-certificate log.

\newpage

\bibliographystyle{plain}
\bibliography{bft}

\begin{thebibliography}{10}

\bibitem{revisit-BFT}
Ittai Abraham, Guy Gueta, Dahlia Malkhi, Lorenzo Alvisi, Rama Kotla, and
  Jean-Philippe Martin.
\newblock Revisiting fast practical byzantine fault tolerance.
\newblock ArXiv, https://arxiv.org/abs/1712.01367, 2017.

\bibitem{sevenBFT}
Pierre-Louis Aublin, Rachid Guerraoui, Nikola Kne\v{z}evi\'{c}, Vivien
  Qu{\'e}ma, and Marko Vukoli\'{c}.
\newblock The next 700 bft protocols.
\newblock {\em ACM Trans. Comput. Syst.}, 32(4):12:1--12:45, January 2015.

\bibitem{CL99}
Miguel Castro and Barbara Liskov.
\newblock Practical byzantine fault tolerance.
\newblock In {\em Proceedings of the Third Symposium on Operating Systems
  Design and Implementation}, OSDI '99, pages 173--186, Berkeley, CA, USA,
  1999. USENIX Association.

\bibitem{CL02}
Miguel Castro and Barbara Liskov.
\newblock Practical byzantine fault tolerance and proactive recovery.
\newblock {\em ACM Trans. Comput. Syst.}, 20(4):398--461, November 2002.

\bibitem{DLS_jacm88}
Cynthia Dwork, Nancy Lynch, and Larry Stockmeyer.
\newblock Consensus in the presence of partial synchrony.
\newblock {\em J. ACM}, 35(2):288--323, April 1988.

\bibitem{RQS}
Rachid Guerraoui and Marko Vukoli\'{c}.
\newblock Refined quorum systems.
\newblock {\em Distributed Computing}, 23(1):1--42, 2010.

\bibitem{K07}
Ramakrishna Kotla, Lorenzo Alvisi, Mike Dahlin, Allen Clement, and Edmund Wong.
\newblock Zyzzyva: Speculative byzantine fault tolerance. \textbf{Best paper
  award}.
\newblock In {\em Proceedings of Twenty-first ACM SIGOPS Symposium on Operating
  Systems Principles}, SOSP '07, pages 45--58, New York, NY, USA, 2007. ACM.

\bibitem{KAD09}
Ramakrishna Kotla, Lorenzo Alvisi, Mike Dahlin, Allen Clement, and Edmund Wong.
\newblock Zyzzyva: Speculative byzantine fault tolerance.
\newblock {\em ACM Trans. Comput. Syst.}, 27(4):7:1--7:39, January 2010.

\bibitem{K08}
Ramakrishna Kotla, Allen Clement, Edmund Wong, Lorenzo Alvisi, and Mike Dahlin.
\newblock Zyzzyva: Speculative byzantine fault tolerance.
\newblock {\em Commun. ACM}, 51(11):86--95, November 2008.

\bibitem{paxos}
Leslie Lamport.
\newblock The part-time parliament.
\newblock {\em ACM Trans. Comput. Syst.}, 16:133--169, May 1998.

\bibitem{MR_dc98}
Dahlia Malkhi and Michael Reiter.
\newblock Byzantine quorum systems.
\newblock {\em Distrib. Comput.}, 11(4):203--213, October 1998.

\bibitem{MA05}
Jean-Philippe Martin.
\newblock Fast byzantine consensus. \textbf{Paper award}.
\newblock In {\em Proceedings of the 2005 International Conference on
  Dependable Systems and Networks}, DSN '05, pages 402--411, Washington, DC,
  USA, 2005. IEEE Computer Society.

\bibitem{MA06}
Jean-Philippe Martin and Lorenzo Alvisi.
\newblock Fast byzantine consensus.
\newblock {\em IEEE Trans. Dependable Secur. Comput.}, 3(3):202--215, July
  2006.

\bibitem{oki_podc88}
Brian~M. Oki and Barbara~H. Liskov.
\newblock Viewstamped replication: A new primary copy method to support
  highly-available distributed systems.
\newblock In {\em Proceedings of the Seventh Annual ACM Symposium on Principles
  of Distributed Computing}, PODC '88, pages 8--17, New York, NY, USA, 1988.
  ACM.

\bibitem{ongaro_atc14}
Diego Ongaro and John Ousterhout.
\newblock In search of an understandable consensus algorithm.
\newblock In {\em Proc. USENIX Annual Technical Conference}, pages 305--320,
  2014.

\end{thebibliography}

\end{document}